\documentclass{article}

\usepackage[utf8]{inputenc} 
\usepackage[T1]{fontenc}    
\usepackage{hyperref}       
\usepackage{url}            
\usepackage{booktabs}       
\usepackage{amsfonts}       
\usepackage{nicefrac}       
\usepackage{microtype}      
\usepackage{amsthm,amsmath,amssymb}
\usepackage{algorithm,algpseudocode}
\usepackage{booktabs}
\usepackage{comment}
\usepackage[sort]{cite}
\usepackage{wrapfig}
\usepackage{enumerate}

\usepackage[skip=2pt]{caption}

\usepackage{pgfplots}
\pgfplotsset{compat=newest}
\usepackage{tikz}
\usetikzlibrary{patterns}

\theoremstyle{definition}
\newtheorem{theorem}              {Theorem}
\newtheorem{lemma}      [theorem] {Lemma}
\newtheorem{proposition}[theorem] {Proposition}

\newtheorem{definition} [theorem] {Definition}
\newtheorem{example}    [theorem] {Example}

\numberwithin{equation}{section}
\numberwithin{figure}{section}
\numberwithin{table}{section}

\newcommand{\floor}[1]{\lfloor #1 \rfloor}
\newcommand{\argmax}{\operatornamewithlimits{argmax}}

\newcommand{\rank}[1]{\mathrm{rank}(#1)}
\newcommand{\adm}[1]{\mathrm{adm}(#1)}
\newcommand{\cl}[1]{\mathrm{cl}(#1)}

\usepackage{color}
\usepackage{CJKutf8}
\usepackage{ifthen}
\newcommand{\COMM}[2]{{
\begin{CJK}{UTF8}{ipxm}
\ifthenelse{\equal{#1}{SN}}{\color{blue}}{
\ifthenelse{\equal{#1}{TM}}{\color{red}}{
\ifthenelse{\equal{#1}{AA}}{\color{cyan}}{
\ifthenelse{\equal{#1}{BB}}{\color{magenta}}}}}
[#1: #2]
\end{CJK}
}}

\title{Subspace Selection via DR-Submodular Maximization on Lattices}

\usepackage[preprint,nonatbib]{nips_2018}
\author{
  So Nakashima  \\
  University of Tokyo \\
  RIKEN Center for Advanced Intelligence Project \\
  \texttt{so\_nakashima@mist.i.u-tokyo.ac.jp } \\
  \And 
  Takanori Maehara \\ 
  RIKEN Center for Advanced Intelligence Project \\
  \texttt{takanori.maehara@riken.jp}
}

\begin{document}

\maketitle

\begin{abstract}
The subspace selection problem seeks a subspace that maximizes an objective function under some constraint. 
This problem includes several important machine learning problems such as the principal component analysis and sparse dictionary selection problem.
Often, these problems can be solved by greedy algorithms. 
Here, we are interested in why these problems can be solved by greedy algorithms, and what classes of objective functions and constraints admit this property.
To answer this question, we formulate the problems as optimization problems on \emph{lattices}. 
Then, we introduce a new class of functions, \emph{directional DR-submodular functions}, to characterize the approximability of problems.
We see that the principal component analysis, sparse dictionary selection problem, and these generalizations have directional DR-submodularities.
We show that, under several constraints, the directional DR-submodular function maximization problem can be solved efficiently with provable approximation factors.
\end{abstract}

\section{Introduction}
\label{sec:introduction}

\paragraph{Background and motivation}

\emph{The subspace selection problem} involves seeking a good subspace from data.
Mathematically, the problem is formulated as follows.
Let $\mathcal{L}$ be a family of subspaces of $\mathbb{R}^d$, $\mathcal{F} \subseteq \mathcal{L}$ be a set of feasible subspaces, and $f \colon \mathcal{L} \to \mathbb{R}$ be an objective function.
Then, the task is to solve the following optimization problem.
\begin{align}
\label{eq:problem}
\begin{array}{ll}
\text{maximize} & f(X) \\
\text{subject to} & X \in \mathcal{F}.
\end{array}
\end{align}
This problem is a kind of feature selection problem, and contains several important machine learning problems such as the principal component analysis and sparse dictionary selection problem.

In general, the subspace selection problem is a non-convex continuous optimization problem; hence it is hopeless to obtain a provable approximate solution.
On the other hand, such solution can be obtained efficiently in some special cases.
The most important example is the principal component analysis. 
Let $\mathcal{L}(\mathbb{R}^d)$ be the set of all the subspaces of $\mathbb{R}^d$, 
$\mathcal{F}$ be the subspaces with dimension of at most $k$, and $f \colon \mathcal{L} \to \mathbb{R}$ be the function defined by
\begin{align}
\label{eq:pca}
	f(X) = \sum_{i \in I} \| \Pi_X u_i \|^2 
\end{align}
where $\{ u_i \}_{i \in I} \subset \mathbb{R}^d$ is the given data and $\Pi_X$ is the projection to subspace $X$.
Then, problem~\eqref{eq:problem} with these $\mathcal{L}(\mathbb{R}^d)$, $\mathcal{F}$, and $f$ defines the principal component analysis problem.
As we know, the greedy algorithm, which iteratively selects a new direction $a_i \in \mathbb{R}^d$ that maximizes the objective function, gives the optimal solution to problem~\eqref{eq:problem}.
Another important problem is the sparse dictionary selection problem.
Let $V \subseteq \mathbb{R}^d$ be a set of vectors, called a dictionary. 
For a subset $S \subseteq V$, we denote by $\mathrm{span}(S)$ the subspace spanned by $S$.
Let $\mathcal{L}(V) = \{ \mathrm{span}(S) : S \subseteq V \}$ be the subspaces spanned by a subset of $V$, and $\mathcal{F}$ be the subspaces spanned by at most $k$ vectors of $V$. 
Then, the problem~\eqref{eq:problem} with these $\mathcal{L}(V)$, $\mathcal{F}$, and $f$ in \eqref{eq:pca} defines the sparse dictionary selection problem.
The problem is in general difficult to solve~\cite{natarajan1995sparse}; however, the greedy-type algorithms, e.g., \emph{orthogonal matching pursuit}, yield provable approximation guarantees depending on the mutual coherence of $V$.

Here, we are interested in the following research question: Why the principal component analysis and the sparse dictionary selection problem can be solved by the greedy algorithms, and what classes of objective functions and constraints have the same property?

\paragraph{Existing approach}

Several researchers have considered this research question (see Related work below). 
One successful approach is employing \emph{submodularity}. 
Let $V \subseteq \mathbb{R}^d$ be a (possibly infinite) set of vectors. 
We define $F \colon 2^V \to \mathbb{R}$ by $F(S) = f(\mathrm{span}(S))$.
If this function satisfies the submodularity, $F(S) + F(T) \ge F(S \cup T) + F(S \cap T)$, or some its approximation variants, we obtain a provable approximation guarantee of the greedy algorithm~\cite{krause2010submodular,das2011submodular,elenberg2016restricted,khanna2017approximation}. 

However, this approach has a crucial issue that it cannot capture the structure of vector spaces.
Consider three vectors $a = (1,0)$, $b = (1/\sqrt{2}, 1/\sqrt{2})$, and $c = (0,1)$ in $\mathbb{R}^2$.
Then, we have 
$\mathrm{span}(\{a, b\}) = \mathrm{span}(\{b, c\}) = \mathrm{span}(\{c, a\})$; therefore, $F(\{a, b\}) = F(\{b, c\}) = F(\{c, a\})$. 
However, this property (a single subspace is spanned by different bases) is overlooked in the existing approach, which yields underestimation of the approximation factors of the greedy algorithms (see Section~\ref{sec:sparse}).

\paragraph{Our approach}

In this study, we employ \emph{Lattice Theory} to capture the structure of vector spaces. 
A \emph{lattice} $\mathcal{L}$ is a partially ordered set closed under the greatest lower bound (aka., meet, $\land$) and the least upper bound (aka., join, $\lor$).

The family of all subspaces of $\mathbb{R}^d$ is called \emph{the vector lattice} $\mathcal{L}(\mathbb{R}^d)$, which forms a lattice whose meet and join operators correspond to the intersection and direct sum of subspaces, respectively.
This lattice can capture the structure of vector spaces as mentioned above.
Also, the family of subspaces $\mathcal{L}(V)$ spanned by a subset of $V \subseteq \mathbb{R}^d$ forms a lattice.

We want to establish a submodular maximization theory on lattice.
Here, the main difficulty is a ``nice'' definition of submodularity.
Usually, the \emph{lattice submodularity} is defined by the following inequality~\cite{topkis1978minimizing}, which is a natural generalization of set submodularity.
\begin{align}
\label{eq:submodular}
	f(X) + f(Y) \ge f(X \land Y) + f(X \lor Y).
\end{align}
However, this is too strong that it cannot capture the principal component analysis as shown below.
\begin{example}
\label{ex:critical}
Consider the vector lattice $\mathcal{L}(\mathbb{R}^2)$.
Let $X = \mathrm{span} \{ (1, 0) \}$ and $Y = \mathrm{span} \{ (1, \epsilon) \}$ be subspaces of $\mathbb{R}^2$ where $\epsilon > 0$ is sufficiently small.
Let $\{ v_i \}_{i \in I} = \{ (0, 1) \}$ be the given data.
Then, function \eqref{eq:pca} satisfies $f(X) = 0$, $f(Y) = \epsilon/\sqrt{1 + \epsilon^2}$, $f(X \land Y) = 0$, and $f(X \lor Y) = 1$. 
Therefore, it does not satisfy the lattice submodularity. 
A more important point is that, since we can take $\epsilon \to 0$, there is no constants $\alpha > 0$ and $\delta \ll f(X) + f(Y)$ such that $f(X) + f(Y) \ge \alpha( f(X \land Y) + f(X \lor Y) ) - \delta$ on this lattice.
This means that it is very difficult to formulate this function as an approximated version of a lattice submodular function.
\end{example}
Another commonly used submodularity is the \emph{diminishing return (DR)-submodularity}~\cite{soma2015generalization,bian2016guaranteed,soma2017non}, which is originally introduced on the integer lattice $\mathbb{Z}^V$.
A function $f \colon \mathbb{Z}^V \to \mathbb{R}$ is DR-submodular if 
\begin{align}
\label{eq:drsubmodular}
	f(X + e_i) - f(X) \ge f(Y + e_i) - f(Y)
\end{align}
for all $X \le Y$ (component wise inequality) and $i \in V$, where $e_i$ is the $i$-th unit vector.
This definition is later extended to distributive lattices~\cite{gottschalk2015submodular} and can be extended to general lattices (see Section~\ref{sec:submodular}). 
However, Example~\ref{ex:critical} above is still crucial, and therefore the objective function of the principal component analysis cannot be an approximated version of a DR-submodular function.

To summarize the above discussion, our main task is to define submodularity on lattices that should satisfy the following two properties:
\begin{enumerate} 
\item It captures some important practical problems such as the principal component analysis.
\item It admits efficient approximation algorithms on some constraints.
\end{enumerate}

\paragraph{Our contributions}

In this study, in response to the above two requirements, we make the following contributions:
\begin{enumerate} 
\item 
We define \emph{downward DR-submodularity} and \emph{upward DR-submodularity} on \emph{lattices}, which generalize the DR-submodularity (Section~\ref{sec:submodular}).
Our directional DR-submodularities are capable of representing important machine learning problems such as the principal component analysis and sparse dictionary selection problem (Section~\ref{sec:examples}).

\item 
We propose approximation algorithms for maximizing (1) monotone downward DR-submodular function over height constraint, (2) monotone downward DR-submodular function over knapsack constraint, and (3) non-monotone DR-submodular function (Section~\ref{sec:algorithms}). 
These are obtained by generalizing the existing algorithms for maximizing the submodular set functions. 
Thus, even our directional DR-submodularities are strictly weaker than the lattice DR-submodularity; it is sufficient to admit approximation algorithms.
\end{enumerate}

All the proofs of propositions and theorems are given in Appendix in the supplementary material.

\paragraph{Related Work}

For the principal component analysis, we can see that the greedy algorithm, which iteratively selects the largest eigenvectors of the correlation matrix, solves the principal component analysis problem exactly~\cite{abdi2010principal}.

With regard to the sparse dictionary selection problem, several studies~\cite{gilbert2003approximation,tropp2003improved,tropp2004greed,das2008algorithms} have analyzed greedy algorithms.
In general, the objective function for the sparse dictionary selection problem is not submodular. 
Therefore, researchers introduced approximated versions of the submodularity and analyzed the approximation guarantee of algorithms with respect to the parameter.

Krause and Cevher~\cite{krause2010submodular} showed that function \eqref{eq:pca} is an \emph{approximately submodular} function whose additive gap $\delta \ge 0$ depends on the mutual coherence. They also showed that the greedy algorithm gives $(1 - 1/e, k \delta)$-approximate solution.%
\footnote{A solution $X$ is an \emph{$(\alpha, \delta)$-approximate solution} if it satisfies $f(X) \ge \alpha \max_{X' \in \mathcal{F}} f(X') - \delta$. If $\delta = 0$ then we simply say that it is an $\alpha$-approximate solution.}

Das and Kempe~\cite{das2011submodular} introduced the \emph{submodularity ratio}, which is another measure of submodularity.
For the set function maximization problem, the greedy algorithm attains a provable approximation guarantee depending on the submodularity ratio.
The approximation ratio of the greedy algorithm is further improved by combining with the curvature~\cite{bian2017guarantees}.
Elenberg et al.~\cite{elenberg2016restricted} showed that, if function $l \colon \mathbb{R}^d \to \mathbb{R}$ has a bounded restricted convexity and a bounded smoothness, the corresponding set function $F(S) := l(0) - \min_{\mathrm{supp}(x) \in S} l(x)$ has a bounded submodularity ratio.
Khanna et al.~\cite{khanna2017approximation} applied the submodularity ratio for the low-rank approximation problem.

It should be emphasized that all the existing studies analyzed the greedy algorithm as a function of a set of vectors (the basis of the subspace), instead of as a function of a subspace. 
This overlooks the structure of the subspaces causing difficulties as described above.

\section{Preliminaries}
\label{sec:preliminaries}

A \emph{lattice} $(\mathcal{L}, \le)$ is a partially ordered set (poset) such that, for any $X, Y \in \mathcal{L}$, the least upper bound $X \lor Y := \inf \{ Z \in \mathcal{L} : X \le Z, Y \le Z \} $ and the greatest lower bound $X \land Y := \sup \{ Z \in L : Z \le X, Z \le Y \} $ uniquely exist.
We often say ``$\mathcal{L}$ is a lattice'' by omitting $\le$ if the order is clear from the context.
In this paper, we assume that the lattice has the smallest element $\bot \in \mathcal{L}$.

A subset $\mathcal{I} \subseteq \mathcal{L}$ is \emph{lower set} if $Y \in \mathcal{I}$ then any $X \in \mathcal{L}$ with $X \le Y$ is also $X \in \mathcal{I}$. 
For $Y \in \mathcal{L}$, the set $\mathcal{I}(Y) = \{ X \in \mathcal{L} : X \le Y \}$ is called the \emph{lower set of $Y$}.

A sequence $X_1 < \cdots < X_k$ of elements of $\mathcal{L}$ is a \emph{composition series} if there is no $Y \in \mathcal{L}$ such that $X_i < Y < X_{i+1}$ for all $i$.
The length of the longest composition series from $\bot$ to $X$ is referred to as the \emph{height} of $X$ and is denoted by $h(X)$.
The height of a lattice is defined by $\sup_{X \in \mathcal{L}} h(X)$.
If this value is finite, the lattice has the largest element $\top \in \mathcal{L}$. 
Note that the height of a lattice can be finite even if the lattice has infinitely many elements. For example, the height of the vector lattice $\mathcal{L}(\mathbb{R}^d)$ is $d$.

A lattice $\mathcal{L}$ is \emph{distributive} if it satisfies the distributive law: $  (X \land Y) \lor Z = (X \lor Z) \land (Y \lor Z) $.
A lattice $\mathcal{L}$ is \emph{modular} if it satisfies the modular law: $ X \le B \Rightarrow X \lor (A \land B) = (X \lor A) \land B $.
Every distributive lattice is modular.
On a modular lattice $\mathcal{L}$, all the composition series between $X \in \mathcal{L}$ and $Y \in \mathcal{L}$ have the same length.
The lattice is modular if and only if its height function satisfies the modular equality: $h(X) + h(Y) = h(X \lor Y) + h(X \land Y)$.
Modular lattices often appear with algebraic structures. 
For example, the set of all subspaces of a vector space forms a modular lattice.
Similarly, the set of all normal subgroups of a group forms a modular lattice. 

For a lattice $\mathcal{L}$, an element $a \in \mathcal{L}$ is \emph{join-irreducible} if there no $X \neq a, Y \neq a$ such that $a = X \lor Y$.\footnote{For the set lattice $2^V$ of a set $V$, the join-irreducible elements correspond to the singleton sets, $\{ a \in V \}$. Thus, for clarity, we use upper case letters for general lattice elements (e.g., $X$ or $Y$) and lower case letters for join-irreducible elements (e.g., $a$ or $b$).}
We denote by $J(\mathcal{L}) \subseteq \mathcal{L}$ the set of all join-irreducible elements.
Any element $X \in \mathcal{L}$ is represented by a join of join-irreducible elements; therefore the structure of $\mathcal{L}$ is specified by the structure of $J(\mathcal{L})$.
A join irreducible element $a \in J(\mathcal{L})$ is \emph{admissible} with respect to an element $X \in \mathcal{L}$ if $a \not \le X$ and any $a' \in \mathcal{L}$ with $a' < a$ satisfies $a' \le X$.
We denote by $\adm{X}$ the set of all admissible elements with respect to $X$.
A set $\cl{a \mid X} = \{ a' \in \adm{X} : X \lor a = X \lor a' \}$ is called a \emph{closure} of $a$ at $X$.
See Figures~\ref{fig:admissible} and \ref{fig:closure} for the definition of admissible elements and closure.
Note that $a$ is admissible with respect to $X$ if and only if the distance from the lower set of $X$ to $a$ is one.

\begin{example}
In the vector lattice $\mathcal{L}(\mathbb{R}^d)$, each element corresponds to a subspace. An element $a \in \mathcal{L}(\mathbb{R}^d)$ is join-irreducible if and only if it has dimension one.
A join-irreducible element $a \in \mathcal{L}(\mathbb{R}^d)$ is admissible to $X \in \mathcal{L}(\mathbb{R}^d)$ if these are linearly independent.
The closure $\cl{a | X}$ is the one dimensional subspaces contained in $X \lor a$ independent to $X$.
\end{example}

\begin{figure}[tb]
\begin{minipage}{0.49\textwidth}
\center
\begin{tikzpicture}[scale=0.70]
\fill[pattern = dots] (0, -0.6) -- (0.6, 0) -- (-2, 2.6) -- (-2.6,2) -- cycle;
\node[circle,draw,inner  sep=1mm,label=left:$\bot$] (bot) at (0,0) {};
\node[circle,draw,inner sep=1mm,label=left:$X$] (X) at (-2,2) {};
\node[circle,draw,inner sep=1mm,label=left:$a$] (a) at (1,1) {};
\node[circle,draw,inner sep=1mm,label=left:$b$] (b) at (2,2) {};
\node[circle,draw,inner sep=1mm,label={left:$X \lor a$}] (Xa) at (-1,3) {};
\node[circle,draw,inner sep=1mm,label={left:$X \lor b$}] (Xb) at (0,4) {};
\draw[-] (bot) -- (X);
\draw[-] (bot) -- (a);
\draw[-] (a) -- (b);
\draw[-] (b) -- (Xb);
\draw[-] (a) -- (Xa);
\draw[-] (X) -- (Xa);
\draw[-] (Xa) -- (Xb);
\end{tikzpicture}
\captionsetup{width=0.95\textwidth}
\caption{$a$ is admissible with respect to $X$ but $b$ is not because of the existence of $a$. The shaded area represents the lower set of $X$.}
\label{fig:admissible}
\end{minipage}
\begin{minipage}{0.49\textwidth}
\center
\begin{tikzpicture}[scale=0.70]
\node[circle,draw,inner sep=1mm,label=left:$\bot$] (bot) at (0,0) {};
\node[circle,draw,inner sep=1mm,label=left:$X$] (X) at (-2,2) {};
\node[circle,draw,inner sep=1mm,label=left:$a$] (a) at (0,2) {};
\node[circle,draw,inner sep=1mm,label=left:$b$] (b) at (2,2) {};
\node[circle,draw,inner sep=1mm,label={left:$X \lor a = X \lor b$}] (top) at (0,4) {};
\draw[-] (bot) -- (X);
\draw[-] (bot) -- (a);
\draw[-] (bot) -- (b);
\draw[-] (X) -- (top);
\draw[-] (a) -- (top);
\draw[-] (b) -- (top);
\end{tikzpicture}
\captionsetup{width=0.95\textwidth}
\caption{Both $a$ and $b$ are admissible with respect to $X$, and $X \lor a = X \lor b$. Thus $b \in \cl{a | X}$ and $a \in \cl{b | X}$.}
\label{fig:closure}
\end{minipage}
\end{figure}

\section{Directional DR-submodular functions on modular lattices}
\label{sec:submodular}

We introduce new submodularities on lattices. 
As described in Section~\ref{sec:introduction}, our task is to find useful definitions of ``submodularities'' on lattices; thus, this section is the most important part of this paper.

Recall definition~\eqref{eq:drsubmodular} of the DR-submodularity on the integer lattice.
Then, we can see that $X + e_i = X \lor a$ and $Y + e_i = Y \lor b$ for $a = (X_i + 1) e_i$ and $b = (Y_i + 1) e_i$, where $X_i$ and $Y_i$ are the $i$-th components of $X$ and $Y$, respectively.
Here, $a$ and $b$ are join-irreducibles in the integer lattice, $a \in \adm{X}$, $b \in \adm{Y}$, and $a \le b$.
Thus, a natural definition of the DR-submodularity on lattices is as follows.
\begin{definition}[Strong DR-submodularity]
A function $f \colon \mathcal{L} \to \mathbb{R}$ is \emph{strong DR-submodular} if, for all $X, Y \in \mathcal{L}$ with $X \le Y$ and $a \in \adm{X}, b \in \adm{Y}$ with $a \le b$, the following holds.
\begin{align}
\label{eq:strongdr}
	f(X \lor a) - f(X) \ge f(Y \lor b) - f(Y)
\end{align}
\end{definition}
The same definition is introduced by Gottshalk and Peis~\cite{gottschalk2015submodular} for distributive lattices.
However, this is too strong for our purpose because it cannot capture the principal component analysis; you can check this in Example~\ref{ex:critical}.
Therefore, we need a weaker concept of DR-submodularities.

Recall that $f(Y \lor b) - f(Y) = f(Y \lor b') - f(Y)$ for all $b' \in \cl{b | Y}$. 
Thus, the strong DR-submodularity~\eqref{eq:strongdr} is equivalent to the following.
\begin{align}
	f(Y \lor b) - f(Y) \le \min_{b' \in \cl{b | Y}} \min_{a \in \adm{X}, a \le b'} f(X \lor a) - f(X).
\end{align}
By relaxing the outer $\min$ to $\max$, we obtain the following definition.
\begin{definition}[Downward DR-submodularity]
Let $\mathcal{L}$ be a lattice.
A function $f \colon \mathcal{L} \to \mathbb{R}$ is \emph{downward DR-submodular with additive gap $\delta$}, if  for all $X \le Y$ and $b \in \adm{Y}$, the following holds.
\begin{align}
\label{eq:downwarddr}
f(Y \lor b) - f(Y) \le  \max_{b' \in \cl{b | Y}} \min_{a \in \adm{X}, a \le b'} f(X \lor a) - f(X) + \delta.
\end{align}
\end{definition}
Similarly, the strong DR-submodularity~\eqref{eq:strongdr} is equivalent to the following.
\begin{align}
	f(X \lor a) - f(X) \ge \max_{b \ge a} \max_{\mathring{Y}: b \in \adm{\mathring{Y}}, Y = \mathring{Y} \lor b, X \le \mathring{Y}} f(Y) - f(\mathring{Y}).
\end{align}
By relaxing the inner $\max$ to $\min$, we obtain the following definition.
\begin{definition}[Upward DR-submodularity]
Let $\mathcal{L}$ be a lattice.
$f \colon \mathcal{L} \to \mathbb{R}$ is \emph{upward DR-submodular with additive gap $\delta$}, if for all $X, Y \in \mathcal{L}$ and $a \in \adm{X}$ with $X \lor a \le Y$, the following holds.
\begin{align}
\label{eq:upwarddr}
f(X \lor a) - f(X) \ge \max_{b \ge a} \min_{\mathring{Y}: b \in \adm{\mathring{Y}}, Y = \mathring{Y} \lor b, X \le \mathring{Y}} f(Y) - f(\mathring{Y}) - \delta
\end{align}
\end{definition}
If a function $f$ is both downward DR-submodular with additive gap $\delta$ and upward DR-submodular with additive gap $\delta$, then we say that $f$ is \emph{bidirectional DR-submodular} with additive gap $\delta$.
We say \emph{directional DR-submodularity} to refer these new DR-submodularities.

The strong DR-submodularity implies the bidirectional DR-submodularity, because both downward and upward DR-submodularities are relaxations of the strong DR-submodularity.
Interestingly, the converse also holds in distributive lattices.
\begin{proposition}
\label{prop:upward_implies_strong}
On a distributive lattice, 
the strong DR-submodularity, downward DR-submodularity, and upward DR-submodularity are equivalent.
\qed
\end{proposition}
Therefore, we can say that directional DR-submodularities are required to capture the specialty of non-distributive lattices such as the vector lattice.

At the cost of generalization, in contrast to the lattice submodularity~\eqref{eq:submodular} and the strong DR-submodularity~\eqref{eq:strongdr}, the downward and upward DR-submodularity are \emph{not} closed under addition, because the elements attained in the min/max in the above definitions can depend on the objective function.
\section{Examples}
\label{sec:examples}

In this section, we present several examples of directional DR-submodular functions to show that our concepts \emph{can} capture several machine learning problems.

\subsection{Principal component analysis}
\label{sec:pca}

Let $\{ u_i \}_{i \in I} \subset \mathbb{R}^d$ be the given data.
We consider the vector lattice $\mathcal{L}(\mathbb{R}^d)$ of all the subspaces of $\mathbb{R}^d$, and the objective function $f$ defined by \eqref{eq:pca}.
Then, the following holds. 

\begin{proposition}
\label{prop:pca}
The function $f \colon \mathcal{L}(\mathbb{R}^d) \to \mathbb{R}$ defined by \eqref{eq:pca} is a monotone bidirectional DR-submodular function.
\qed
\end{proposition}
This provides a reason why the principal component analysis is solved by the greedy algorithm from the viewpoint of submodularity.

The objective function can be generalized further. 
Let $\rho_i: \mathbb{R} \to \mathbb{R}$ be a monotone non-decreasing concave function with $\rho_i(0) = 0$ for each $i \in I$.
Let
\begin{align}
\label{eq:generalizedpca}
	f_\rho(X) = \sum_{i \in I} \rho_i( \| \Pi_X u_i \|^2 ).
\end{align}
Then, the following holds.
\begin{proposition}
\label{prop:generalizedpca}
The function $f_\rho \colon \mathcal{L}(\mathbb{R}^d) \to \mathbb{R}$ defined by \eqref{eq:generalizedpca} is a monotone bidirectional DR-submodular function.
\qed
\end{proposition}
If we use this function instead of the standard function~\eqref{eq:pca}, we can ignore the contributions from very large vectors because if $u_i$ is already well approximated in $X$, there is less incentive to seek larger subspace for $u_i$ due to the concavity of $\rho_i$.
See Experiment in Appendix.

\subsection{Sparse dictionary selection}
\label{sec:sparse}

Let $V \subseteq \mathbb{R}^d$ be a set of vectors called a dictionary.
We consider $\mathcal{L}(V) = \{ \mathrm{span}(S) : S \subseteq V \}$ of all subspaces spanned by $V$, which forms a (not necessarily modular) lattice. 
The height of $X \in \mathcal{L}(V)$ coincides with the dimension of $X$.
Let $\{ u_i \}_{i \in I} \subset \mathbb{R}^d$ be the given data.
Then the sparse dictionary selection problem is formulated by the maximization problem of $f$ defined by \eqref{eq:pca} on this lattice under the height constraint.

In general, the function $f$ is not a directional DR-submodular function on this lattice. 
However, we can prove that $f$ is a downward DR-submodular function with a provable additive gap. 
We introduce the following definition.

\begin{definition}[Mutual coherence of lattice]
Let $\mathcal{L}$ be a lattice of subspaces. 
For $\epsilon \ge 0$, the lattice has \emph{mutual coherence $\epsilon$}, if for any $X \in \mathcal{L}$, there exists $X' \in \mathcal{L}$ such that $X \land X' = \bot$, $X \lor X' = \top$, and for all unit vectors $u \in X$ and $u' \in X'$, 
$ |\langle u, u' \rangle| \le \epsilon $.
The infimum of such $\epsilon$ is called the \emph{mutual coherence of $\mathcal{L}$}, and is denoted by $\mu(\mathcal{L})$.
\end{definition}
Our mutual coherence of a lattice is a generalization of the \emph{mutual coherence} of a set of vectors~\cite{donoho2003optimally}. 
For a set of unit vectors $V = \{u_1, \ldots, u_N\} \subset \mathbb{R}^d$, its mutual coherence is defined by $\mu(V) = \max_{i \neq j} |\langle u_i, u_j \rangle|$.
The mutual coherence of a set of vector is extensively used in compressed sensing to prove the uniqueness of the solution in a sparse recovery problem~\cite{eldar2012compressed}.
Here, we have the following relation between the mutual coherence of a lattice and that of a set of vectors, which is the reason why we named our quantity mutual coherence.
\begin{lemma}
\label{lem:mutualcoherence}
Let $V = \{ u_1, \ldots, u_N \}$ be a set of unit vectors whose mutual coherence is $\mu(V) \le \epsilon$.
Then, the lattice $\mathcal{L}(V)$ generated by the vectors has mutual coherence $\mu(\mathcal{L}(V)) \le d \epsilon / (1 - d \epsilon)$.
\qed
\end{lemma}
This means that if a set of vectors has a small mutual coherence, then the lattice generated by the vectors has a small mutual coherence.
Note that the converse does not hold. 
Consider $V = \{ u_1, u_2, u_3 \} \subset \mathbb{R}^2$ where $u_1 = (1,0)^\top$, $u_2 = (1/\sqrt{1+\epsilon^2}, \epsilon/\sqrt{1+\epsilon^2})^\top$, and $u_3 = (0,1)^\top$ for sufficiently small $\epsilon$.
Then the mutual coherence $\mu(V)$ of the vectors is $1/\sqrt{1 + \epsilon^2} \approx 1$; however, the mutual coherence $\mu(\mathcal{L})$ of the lattice generated by $V$ is $\epsilon/\sqrt{1 + \epsilon^2} \approx \epsilon$.
This shows that the mutual coherence of a lattice is a more robust concept than that of a set of vectors, which is a strong advantage of considering a lattice instead of a set of vectors.

If a lattice has a small mutual coherence, we can prove that the function $f$ is a monotone downward DR-submodular function with a small additive gap.
\begin{proposition}
\label{prop:ortho}
Let $V = \{ u_1, \ldots, u_N \} \subseteq \mathbb{R}^d$ be normalized vectors and $\mathcal{L}(V)$ be a lattice generated by $V$.
Suppose that $\mathcal{L}(V)$ forms a modular lattice.
Let $\{ v_i \}_{i \in I} \subset \mathbb{R}^d$.
Then, the function $f$ defined in \eqref{eq:generalizedpca} 
is a downward DR-submodular function with additive gap at most $3 \epsilon\rho(0) \sum_j \| v_j \|^2 / (1 - \epsilon^2)$ where $\epsilon = \mu(\mathcal{L}(V))$. 
\qed
\end{proposition}

\subsection{Quantum cut}

Finally, we present an example of a non-monotone bidirectional DR-submodular function.
Let $G = (V, E)$ be a directed graph, and $c: E \to \mathbb{R}_{\ge 0}$ be a weight function.
The cut function is then defined by
$ g(S) = \sum_{(i,j) \in E} c(i,j) 1[ i \in S ] 1[ j \in \bar{S} ] $
where $1[i \in S]$ is the indicator function of $i \in S$ and $\bar{S}$ is the complement of $S$.
This is a non-monotone submodular function.
Maximizing the cut function has application in feature selection problems with diversity~\cite{lin2009graph}.

We extend the cut function to the ``quantum'' setting.
We say that a lattice of vector spaces $\mathcal{L}$ is \emph{ortho-complementable} if $X \in \mathcal{L}$ then $\bar{X} \in \mathcal{L}$ where $\bar{X}$ is the orthogonal complement of $X$.
Let $\{ u_i \}_{i \in V} \subset \mathbb{R}^d$ be vectors assigned on each vertex.
For an ortho-complementable lattice $\mathcal{L}$, the \emph{quantum cut function} $f \colon \mathcal{L} \to \mathbb{R}$ is defined by
\begin{align}
\label{eq:qcut}
	f(X) = \sum_{(i,j) \in E} c(i,j) \| \Pi_X(u_i) \|^2 \| \Pi_{\bar{X}}(v_j) \|^2 .
\end{align}
If $u_i = e_i \in \mathbb{R}^V$ for all $i$, where $e_i$ is the $i$-th unit vector, and $\mathcal{L}$ is the lattice of axis-parallel subspaces of $\mathbb{R}^V$, function~\eqref{eq:qcut} coincides with the original cut function.
Moreover, it carries the submodularity.
\begin{proposition}
\label{prop:qcut}
The function $f$ defined by \eqref{eq:qcut} is a bidirectional DR-submodular function.
\qed
\end{proposition}
The quantum cut function could be used for subspace selection problem with diversity. 
For example, in a natural language processing problem, the words are usually embedded into a latent vector space $\mathbb{R}^d$~\cite{mikolov2013distributed}.
Usually, we select a subset of words to summarize documents; however, if we want to select a ``meaning'', which is encoded in the vector space as a subspace~\cite{kim2013deriving}, it would be promising to select a subspace. 
In such an example, the quantum cut function~\eqref{eq:qcut} can be used to incorporate the diversity represented by the graph of words.

\section{Algorithms}
\label{sec:algorithms}

We provide algorithms for maximizing (1) a monotone downward-DR submodular function on the height constraint, which generalizes the cardinality constraint (Section~\ref{sec:cardinality}), (2) a monotone downward DR-submodular function on knapsack constraint (Section~\ref{sec:knapsack}), and (3) a non-monotone bidirectional DR-submodular function (Section~\ref{sec:nonmonotone}).
Basically, these algorithms are extensions of the algorithms for the set lattice. This indicates that our definitions of directional DR-submodularities are natural and useful.

Below, we always assume that $f$ is normalized, i.e., $f(\bot) = 0$.

\subsection{Height constraint}
\label{sec:cardinality}

\begin{figure}
\begin{minipage}[t]{0.48\textwidth}
\begin{algorithm}[H]
\caption{Greedy algorithm for monotone height constrained problem.}
\label{alg:greedy}
\begin{algorithmic}[1]
\State{$X = \bot$}
\For{$i = 1, \ldots, k$}
\State{Let $a_i \in \displaystyle \argmax_{a \in \adm{X}, X \lor a \in \mathcal{F}} f(X \lor a)$}
\State{$X \leftarrow X \lor a_i$}
\EndFor
\State{\Return $X$}
\end{algorithmic}
\end{algorithm}
\end{minipage} \
\begin{minipage}[t]{0.48\textwidth}
\begin{algorithm}[H]
\caption{Greedy algorithm for monotone knapsack constrained problem.}
\label{alg:knapsack1}
\begin{algorithmic}[1]
\State{$X = \bot$}
\For{$i = 1, 2, \ldots $}
\State{Let $a_i \in  \displaystyle \argmax_{a \in \adm{X}} (f(X \lor a) - f(X))/(c(X \lor a) - c(X))$}
\State{\textbf{if} $c(X \lor a) \le B$ \textbf{then} $X \leftarrow X \lor a_j$}
\EndFor
\State{$a \in \argmax_{a \in \adm{\bot}: c(a) \le B} f(a)$}
\State{\Return $\argmax \{ f(X), f(a) \}$}
\end{algorithmic}
\end{algorithm}
\end{minipage} \ 
\end{figure}

We first consider the height constraint, i.e., $\mathcal{F} = \{ X \in \mathcal{L} : h(X) \le k \}$.
This coincides with the cardinality constraint if $\mathcal{L}$ is the set lattice.
In general, this constraint is very difficult analyze because $h(X \lor a) - h(X)$ can be arbitrary large. 
Thus, we assume that the height function is \emph{$p$-incremental}, i.e., $h(X \lor a) - h(X) \le p$ for all $X$ and $a \in \adm{X}$.
Note that $p = 1$ if and only if $\mathcal{L}$ is modular.

We show that, as similar to the set lattice, the greedy algorithm (Algorithm~\ref{alg:greedy}) achieves $1 - e^{-1/p}$ approximation for the downward DR-submodular maximization problem over the height constraint.
\begin{theorem}
\label{thm:cardinality}
Let $\mathcal{L}$ be a lattice whose height function is $p$-incremental, and $f \colon \mathcal{L} \to \mathbb{R}$ be a downward DR-submodular function with additive gap $\delta$.
Then, Algorithm~\ref{alg:greedy} finds $(1 - e^{-\floor{k/p}/k}, \delta (1 - e^{-\floor{k/p}/k}) k)$-approximate solution of the height constrained monotone submodular maximization problem.%
\footnote{Algorithm~\ref{alg:greedy} requires solving the non-convex optimization problem in Step~3.
If we can only obtain an $\alpha$-approximate solution in Step~3, the approximation ratio of the algorithm reduces to $(1 - e^{\alpha\floor{k/p}/k)}, \delta (1 - e^{\alpha\floor{k/p}/k}) k)$.}
In particular, on modular lattice with $\delta = 0$, it gives $1 - 1/e$ approximation.
\qed
\end{theorem}

\subsection{Knapsack constraint}
\label{sec:knapsack}

\begin{figure}
\begin{algorithm}[H]
\caption{Double-greedy algorithm for non-monotone unconstrained problem.}
\label{alg:nonmonotone}
\begin{algorithmic}[1]
\State{$A = \bot$, $B = \top$}
\While{$A \neq B$}
\State{$\mathring{B} \gets \mathrm{argmin}_{\mathring{B}} f(B) - f(\mathring{B})$ where $\mathring{B}$ runs over $A < \mathring{B} < B$ and $h(\mathring{B}) + 1 = h(B)$}
\State{$a \gets \argmax_{a \in \adm{A}, a \le B} f(A \lor a) - f(A)$} 
\State{\textbf{if} $f(A \lor a) - f(A) \ge f(\mathring{B}) - f(B)$ \textbf{then} $A \leftarrow A \lor a$ \textbf{else} $B \leftarrow \mathring{B}$}
\EndWhile
\State{\textbf{return} $A$}
\end{algorithmic}
\end{algorithm}
\end{figure}

Next, we consider the knapsack constrained problem.  
A knapsack constraint on a lattice is specified by a nonnegative modular function (cost function) $c \colon \mathcal{L} \to \mathbb{R}_{\ge 0}$ and nonnegative number (budget) $B \in \mathbb{R}$ such that the feasible region is given by $\mathcal{F} = \{ X \in \mathcal{L} : c(X) \le B \}$. 

In general, it is NP-hard to obtain a constant factor approximation for a knapsack constrained problem even for a distributive lattice~\cite{gottschalk2015submodular}. 
Therefore, we need additional assumptions on the cost function.

We say that a modular function $c \colon \mathcal{L} \to \mathbb{R}$ is \emph{order-consistent} if $c(X \lor a) - c(X) \le c(Y \lor b) - c(Y)$ for all $X, Y \in \mathcal{L}$, $a \in \adm{X}$, $b \in \adm{Y}$, and $a \le b$.
The height function of a modular lattice is order-consistent, because $c(X \lor a) - c(X) = 1$ for all $X \in \mathcal{L}$ and $a \in \adm{X}$; therefore it generalizes the height function.
Moreover, on the set lattice $2^V$, any modular function is order-consistent because there is no join-irreducible $a, b \in 2^V$ such that $a < b$ holds; therefore it generalizes the standard knapsack constraint on sets.

For a knapsack constraint with an order-consistent nonnegative modular function, we obtain a provable approximation ratio.

\begin{theorem}
\label{thm:knapsack}
Let $\mathcal{L}$ be a lattice, $\mathcal{F} = \{ X \in \mathcal{L} : c(X) \le B \}$ be a knapsack constraint where $c \colon \mathcal{L} \to \mathbb{R}_{\ge 0}$ be an order-consistent modular function, $B \in \mathbb{R}_{\ge 0}$, and $f \colon \mathcal{L} \to \mathbb{R}$ be a monotone downward DR-submodular function with additive gap $\delta$.
Then, Algorithm~\ref{alg:knapsack1} gives  $((1 - e^{-1})/2, \delta h(X^*) (1 - e^{-1})/2)$ approximation of the knapsack constrained monotone submodular maximization problem.
\qed
\end{theorem}

\subsection{Non-monotone unconstrained maximization}
\label{sec:nonmonotone}

Finally, we consider the unconstrained non-monotone maximization problem.

The double greedy algorithm~\cite{buchbinder2015tight} achieves the optimal $1/2$ approximation ratio on the unconstrained non-monotone submodular set function maximization problem.
To extend the double greedy algorithm to lattices, we have to assume that the lattice has a finite height. This is needed to terminate the algorithm in a finite step.
We also assume \emph{both} downward DR-submodularity and upward DR-submodularity, i.e., bidirectional DR-submodularity.
Finally, we assume that the lattice is modular. This is needed to analyze the approximation guarantee.

\begin{theorem}
\label{thm:nonmonotone}
Let $\mathcal{L}$ be a modular lattice of finite height, $\mathcal{F} = \mathcal{L}$, and $f\colon \mathcal{L} \to \mathbb{R}_{\ge 0}$ be non-monotone bidirectional DR-submodular function with additive gap $\delta$. 
Then, Algorithm \ref{alg:nonmonotone} gives $(1/3, \delta h(\mathcal{L}))$ approximate solution of the unconstrained non-monotone submodular maximization problem.
\end{theorem}

\section{Conclusion}
\label{sec:conclusion}

In this paper, we formulated the subspace selection problem as optimization problem over lattices. 
By introducing new ``DR-submodularities'' on lattices, named \emph{directional DR-submodularities}, we successfully characterize the solvable subspace selection problem in terms of the submodularity.
In particular, our definitions successfully capture the solvability of the principal component analysis and sparse dictionary selection problem.
We propose algorithms with provable approximation guarantees for directional DR-submodular functions over several constraints.

There are several interesting future directions.
Developing an algorithm for the matroid constraint over lattice is important since it is a fundamental constraint in submodular set function maximization problem.
Related with this direction, extending the continuous relaxation type algorithms over lattices is very interesting. Such algorithms have been used to obtain the optimal approximation factors to matroid constrained submodular set function maximization problem.

It is also an interesting direction to look for machine learning applications of the directional DR-submodular maximization other than the subspace selection problem.
The possible candidates include the subgroup selection problem and the subpartition selection problem.





\bibliographystyle{plain}
\bibliography{main}


\clearpage

\appendix

\begin{center}
{\LARGE Appendix}
\end{center}

\section{Proofs}
\label{sec:proofs}

In this section, we provide proofs omitted in the main body.

\begin{proof}[Proof of Proposition~\ref{prop:upward_implies_strong}]
We use the Birkhoff's representation theorem for distributive lattice.
A set $A \subseteq J(\mathcal{L})$ is a \emph{lower set} if $b \in A$ then $a \in A$ for all $a \le b$.
The lower sets forms a lattice under the inclusion order. We call this lattice \emph{lower set lattice} of $J(\mathcal{L})$.
\begin{theorem}[Birkhoff's representation theorem; see \cite{gratzer2002general}]
Any finite distributive lattice $\mathcal{L}$ is isomorphic to the lower set lattice of $J(\mathcal{L})$.
The isomorphism is given by $\mathcal{L} \ni X \mapsto \{ a \in J(\mathcal{L}) : a \le X \}$.
\qed
\end{theorem}
This theorem implies that, for any $X \in \mathcal{L}$, the corresponding lower set of $J(\mathcal{L})$ is uniquely determined. 
Therefore, for any $X \in \mathcal{L}$, we have $\cl{a | X} = \{ a \}$ for all $a \in \adm{X}$.

\textbf{(Downward $\Rightarrow$ Strong)} By Birkhoff's representation theorem, we have $\cl{a|X} = \{a\}$. 
Thus the replaced maximum in \eqref{eq:downwarddr} coincides with the minimum.

\textbf{(Upward $\Rightarrow$ Strong)}
By Birkhoff's representation theorem, for any $Y$ and $b \in J(\mathcal{L})$ with $b \le Y$, the element $Y' \in \mathcal{L}$ such that $Y' \lor b = Y$ is uniquely determined (i.e., represent $Y$ as a lower set of $J(\mathcal{L})$ and remove $b$ from the lower set). 
Thus, the replaced minimum in \eqref{eq:upwarddr} coincides with the maximum.
\end{proof}

\begin{proof}[Proofs of Propositions~\ref{prop:pca}, \ref{prop:generalizedpca}]
The downward DR-submodularity follows from Proposition~\ref{prop:ortho}, which is proved below, since the mutual coherence of $\mathcal{L}(\mathbb{R}^d)$ is zero. 
Thus, we here prove the upward DR-submodularity.
To simplify the notation, we prove the case that $f(X) = \rho( \| \Pi_X v\|^2 )$. 
Extension to the general case is easy.

Let $X,Y \in \mathcal{L}$ and $a \in J(\mathcal{L})$ with $X \lor a \le Y$.
Since the height of join-irreducible elements $J(\mathcal{L})$ is one in the vector lattice, the outer max in (\ref{eq:upwarddr}) is negligible.
Let $a' = (X \lor a) \land X^{\perp}$, where $X^{\perp}$ is the orthogonal complement of $X$.
By the modularity of the height, $a'$ is 1-dimensional subspace.
In particular, it is join-irreducible.
Notice that $f(X \lor a) - f(X) = f(X \lor a^{\perp}) - f(X)$.
Since $a' \in X^{\perp}$, we have 
\begin{align}
  f(X \lor a^{\perp}) - f(X) = \rho( \| \Pi_X v\|^2 + \langle a' , v \rangle^2 ) - \rho( \| \Pi_X v\|^2 ).
\end{align}
Here, we identify 1-dimensional subspace $a'$ as a unit vector in the space.
Let $\mathring{Y} = Y \land (a')^{\perp}$.
By using the modularity of the height again, we have $h(\mathring{Y}) + 1 = h(Y)$.
Since $a' \in \mathring{Y}^{\perp}$, we have
\begin{align}
  f(Y) - f(\mathring{Y}) = \rho( \| \Pi_{\mathring{Y}} v\|^2 + \langle a' , v \rangle^2 ) - \rho( \| \Pi_{\mathring{Y}} v\|^2 ).
\end{align}
By the concavity of $\rho$ and $X \subset \mathring{Y}$, we obtain
\begin{align}
  \rho( \| \Pi_X v\|^2 + \langle a' , v \rangle^2 ) - \rho( \| \Pi_X v\|^2 ) \ge \rho( \| \Pi_{\mathring{Y}} v\|^2 + \langle a' , v \rangle^2 ) - \rho( \| \Pi_{\mathring{Y}} v\|^2 ).
\end{align}
This shows the upward DR-submodularity.
\qed


\end{proof}

\begin{proof}[Proofs of Proposition~\ref{prop:ortho}]
%
To simplify the notation, we prove the case that $f(X) = \rho( \| \Pi_X v\|^2 )$. 
Extension to the general case is easy.

Let $\epsilon = \mu(\mathcal{L})$.
Since the join-irreducible elements has height one in this lattice, the additive gap is given by
\begin{align}
	\delta = \sup_{X,Y,b} \left[ f(Y \lor b) - f(Y) - \max_{\tilde b \in \cl{b|Y}, \tilde b \in \adm{X}} f(X \lor \tilde b) - f(X)  \right]
\end{align}
Let $X \le Y$ and $b$ arbitrary.
By the definition of mutual coherence, there exists $Y'$ that has low coherence with $Y$.
Let $b' = (Y \lor b) \land Y'$. 
Then, by the modularity of the height function, we have $b' \neq \bot$ and it is a join-irreducible element. 
Since $b' \le Y \lor b$, we have $Y \lor b' \le Y \lor b$. 
By comparing the height of $Y \lor b'$ and $Y \lor b$, we have $Y \lor b' = Y \lor b$.

We use $b'$ at the RHS and evaluate
\begin{align}
	\left( f(Y \lor b) - f(Y) \right) - \left( f(X \lor b') - f(X) \right).
\end{align}
Let $b^\bot$ be a unit vector in $Y \lor b$ orthogonal to $Y$.
Note that $b^\bot$ may not be the element of $\mathcal{L}$.
\begin{align}
	f(Y \lor b^\bot) - f(Y) 
    &= \rho( \| \Pi_Y v\|^2 + \langle b^\bot, v \rangle^2 ) - \rho( \| \Pi_Y v \|^2 ) \\
    &\le \rho( \| \Pi_X v\|^2 + \langle b^\bot, v \rangle^2 ) - \rho( \| \Pi_X v \|^2 ).
\end{align}
where the second inequality follows from the concavity of $\rho$ with the monotonicity of the mapping $Y \mapsto \| \Pi_Y v \|^2$.
Thus, 
\begin{align}
	\delta \le \rho( \| \Pi_X v \|^2 + \langle b^\bot, v \rangle^2) -  \rho( \| \Pi_X v \|^2 + \langle b'', v \rangle^2).
\end{align}
where $b''$ is the unit vector proportional to $b' - \Pi_X b'$.
If $\langle b^\bot, v \rangle^2 \le \langle b'', v \rangle^2$ then, by the monotonicity of $\rho$, we have $\delta \le 0$.
Therefore, we only have to consider the reverse case.
In such case, by the concavity, we have
\begin{align}
	\delta \le \rho'(0) \left( \langle b^\bot, v \rangle^2 - \langle b'', v \rangle^2 \right).
\end{align}
Here, $\rho'(0)$ is the derivative of $\rho$ at $0$.

Let us denote $b' = \alpha b^\bot + \beta t$ where $t$ is a unit vector in $Y$ orthogonal to $b^\bot$. 
Then, by the definition of the mutual coherence, we have $\beta^2 = \langle b', t \rangle^2 \le \epsilon^2$. 
Also, we have $\alpha^2 = 1 - \beta^2 \ge 1 - \epsilon^2$.
By the construction, we have $b'' = \frac{\alpha b^\bot + \beta \tilde t}{\| \alpha b^\bot + \beta \tilde t\|} = \frac{\alpha b^\bot + \beta \tilde t}{\sqrt{\alpha^2 + \beta^2 \|\tilde t\|^2}}$ where $\tilde t = t - \Pi_X t$.
Thus, we have
\begin{align}
(\alpha^2 + \beta^2 \|\tilde t\|^2) \langle b'', v \rangle^2 
&= (\alpha \langle b^\bot, v \rangle + \beta \langle \tilde t, v \rangle )^2  \\
&= \alpha^2 \langle b^\bot, v \rangle^2 + 2 \alpha \beta \langle b^\bot, v \rangle \langle \tilde t, v \rangle + \beta^2 \langle \tilde t, v \rangle^2.
\end{align}
Therefore, by using $\| \tilde t \|^2 \le \| t \|^2 \le 1$, we have
\begin{align}
\langle b^\bot v \rangle^2 - \langle b'', v \rangle^2 
&\le \frac{2 \alpha \beta \langle b^\bot, v \rangle \langle \tilde t, v \rangle + \beta^2 (\| \tilde t \|^2 \langle b'', v \rangle^2 - \langle \tilde t, v \rangle^2)}{\alpha^2} \\
&\le \frac{2 \alpha \beta + \beta^2}{\alpha^2} \| v \|^2 \le \frac{2 \epsilon \sqrt{1 - \epsilon^2} + \epsilon^2}{1 - \epsilon^2} \|v\|^2 \le \frac{ 3 \epsilon }{1 - \epsilon^2} \|v\|^2.
\end{align}
\end{proof}

\begin{proof}[Proof of Lemma~\ref{lem:mutualcoherence}]
Suppose that $\top \in \mathcal{L}$ has dimension $d$.
Let $X \in \mathcal{L}$. 
Then, there exists $\{ u_{i_1}, \ldots, u_{i_k} \} \subset V$ such that any vector $u \in X$ is represented by a linear combination of them.
We construct $X' \in \mathcal{L}$ by selecting maximally independent vectors to $X$ and let $X' = \mathrm{span}(u_{j_1}) \lor \cdots \lor \mathrm{span}(u_{j_{d-k}})$, where $\{ u_{j_1}, \ldots, u_{j_{d-k}} \} \subset V$.
By the dimension theorem of vector space and the fact that $X \land X'$ is the subspace of the intersection of $X$ and $X'$, we have $\mathrm{dim}(X \land X') + \mathrm{dim}(X \lor X') \le \mathrm{dim}(X) + \mathrm{dim}(X')$. 
Here, the left-hand side is $d + \mathrm{dim}(X \land X')$ and the right-hand side is $k + (d - k) = d$.
Therefore, $\mathrm{dim}(X \land X') = 0$. This shows $X \land X' = \bot$. 

We check the condition of the mutual coherence.
Let $u = \sum_p \alpha_p u_{i_p}$ and $u' = \sum_q \beta_q u_{j_q}$ be normalized vectors in $X$ and $X'$. 
Then we have
\begin{align}
	\langle u, u' \rangle = \alpha^\top M \beta,
\end{align}
where $\alpha = (\alpha_1, \ldots, \alpha_k)^\top$, $\beta = (\beta_1, \ldots, \beta_{d-k})^\top$, and $M_{pq} = \langle u_{i_p}, u_{j_q} \rangle$.
Here, $|\alpha^\top M \beta| \le d \epsilon \|\alpha\| \|\beta\|$.
Therefore we prove that $\| \alpha \|$ and $\| \beta \|$ are small.
Since $u$ is normalized, we have
\begin{align}
	1 = \| u \|^2 = \alpha^\top G \alpha \ge \|\alpha \|^2 \lambda_\text{min}(G)
\end{align}
where $G_{pp'} = \langle u_{i_p}, u_{i_{p'}} \rangle$ and $\lambda_\text{min}(G)$ is the smallest eigenvalue of $G$.
Since the diagonal elements of $G$ are one, and the absolute values of the off-diagonal elements are at most $\epsilon$, the Gerschgorin circle theorem~\cite{strang1993introduction} implies that $\lambda_\text{min}(G) \ge 1 - d \epsilon$. 
Therefore, $\| \alpha \| \le 1/\sqrt{1 - d \epsilon}$. Similarly, $\| \beta \| \le 1/\sqrt{1 - d \epsilon}$.
Therefore, $|\langle u, u' \rangle| \le d \epsilon / (1 - d \epsilon)$.
\end{proof}

\begin{proof}[Proof of Proposition~\ref{prop:qcut}]
We first check the downward DR-submodularity.
Take arbitrary subspaces $X, Y$ and $b$ with $X \le Y$ and $b \not \in Y$.
Without loss of generality, we can suppose $b \in Y^\perp$.
To simplify the notation, we use the same symbol $b$ to represent the unit vector in the subspace $b$. 
By a direct calculation,
\begin{align}
  f(Y \lor b) = &\sum_{(i,j) \in E} c(i,j)(\| \Pi_{Y} v_i\|^2 + \langle b, v_i \rangle^2) ( \| \Pi_{\bar{Y}} v_j \|^2  - \langle b, v_j \rangle^2).
\end{align}
Hence,
\begin{align}
  f(Y \lor b) - f(Y) = &\sum_{(i,j) \in E} c(i,j)( \| \Pi_{\bar{Y}} v_j \|^2 \langle b, v_i \rangle^2 - \| \Pi_Y v_i \|^2 \langle b, v_j \rangle^2 -  \langle b, v_i \rangle^2 \langle b, v_j \rangle^2).
\end{align}
Since $b \in X^\perp$, we have
\begin{align}
  f(X \lor b) - f(X) = &\sum_{(i,j) \in E} c(i,j)( \| \Pi_{\bar{X}} v_j \|^2 \langle b, v_i \rangle^2 - \| \Pi_X v_i \|^2 \langle b, v_j \rangle^2 -  \langle b, v_i \rangle^2 \langle b, v_j \rangle^2).
\end{align}
Since $S \le T$, we have $\| \Pi_{\bar{X}} v_j \|^2 \ge \| \Pi_{\bar{Y}} v_{j}\|^2$ and $ \| \Pi_X v_i \|^2 \le  \| \Pi_Y v_i \|^2$.
Hence, each summand in $f(Y \lor b) - f(Y)$ is smaller than that in $f(X \lor b) - f(X)$.
This shows the downward DR-submodularity.

Next, we check the upward DR-submodularity.
Take arbitrary subspaces $X \le Y$ and a vector $a \in X^\perp$ with $X \lor a \le Y$.
Let $\mathring{Y} = Y \land \mathrm{span}(a)^\perp$.
To simplify the notation, we use the same symbol $a$ to represent the unit vector in the subspace $a$. 
Notice that $X \le \mathring{Y}$.
Then, we can show the following equalities by the same argument as the downward DR-submodular case.
\begin{align}
	f(X \lor a) - f(X) = \sum_{(i,j) \in E} ( \| \Pi_{\bar X} v_j \|^2 \langle a, v_j \rangle^2 - \| \Pi_X v_i \|^2 \langle a, v_i \rangle^2 -  \langle a, v_i \rangle^2 \langle a, v_j \rangle^2), 
\end{align}
\begin{align}
	f(Y) - f(\mathring{Y}) = \sum_{(i,j) \in E} ( \| \Pi_{\bar {\mathring{Y}}} v_j \|^2 \langle a, v_j \rangle^2 - \| \Pi_{\mathring{Y}} v_i \|^2 \langle a, v_i \rangle^2 -  \langle a, v_i \rangle^2 \langle a, v_j \rangle^2), 
\end{align}
By comparing the summand, we have $f(X \lor a) - f(X) \ge f(Y) - f(\mathring{Y})$.
This implies the upward DR-submodularity.
\end{proof}

\begin{proof}[Proof of Proposition~\ref{thm:cardinality}]
Since the increment of the height function is bounded by $p$, the algorithm iterates at least $r = \floor{k / p}$ steps.
For all $i = 1, 2, \dots, r-1$, we can prove the following inequality.
Let $X_i$ be the $X$ in Algorithm~\ref{alg:greedy} after the $i$-the iteration.
The optimal solution is denoted by $X^*$. 
We take $\{a_1, a_2, \dots, a_l\} \subset J(\mathcal{L})$ so that $a_{j+1}$ is admissible to $X^*_j := X_i \lor a_1 \lor a_2 \lor \dots \lor a_{j}$ and $X^*_l = X^* \lor X_i$.
Since $X^* = b_1 \lor b_2 \lor \dots \lor b_{l'}$ for some $l' \le h(X^*)$ and we can take $a_i$ as a subsequence of $b_i$, we obtain $l \le h(X^*) \le k$.
  We set $X^*_0$ to $X_i$.
  Then,
  \begin{align}
    f(X^*) \le f(X^*  \lor X_i) &\le f(X_i) + \sum_{j=1}^l \left[ f(X^*_{j}) - f(X^*_{j-1}) \right]\\
    &\le^{(*)} f(X_i) + \sum_{j=1}^l \left[ f(X_{i+1}) - f(X_{i}) + \delta \right] \\
    &\le f(X_i) + k (f(X_{i+1}) - f(X_{i})) + \delta k.
  \end{align}
  In (*), we used the greediness of our algorithm and DR-submodularity:
  Let $w \in \cl{a_j|X^*_{j-1}}$ and $w' \le w$ such that $w' \in \adm{X_{i-1}}$.
  By greediness of the algorithm, $f(X_{i-1} \lor w') - f(X_{i-1}) \le f(X_i) - f(X_{i-1})$.
  Hence, the downward DR-submodularity implies
  \begin{align}
    f(X^*_{j}) - f(X^*_{j-1}) \le f(X_{i+1}) - f(X_i) + \delta.
  \end{align}
Let $\Delta_i = f(X^*) - f(X_i) - \delta k$.
Then, the above inequality implies
\begin{align}
  \Delta_{i+1} \le (1 - \frac{1}{k}) \Delta_{i}.
\end{align}
Hence,
\begin{align}
  \Delta_{k} \le (1 - \frac{1}{k})^r \delta_0  \le e^{-r/k} \Delta_0.
\end{align}
Therefore,
\begin{align}
	f(X) \ge f(X_r) \ge (1 - e^{-\floor{k/p}/k}) f(X^*) - \delta (1 - e^{-\floor{k/p}/k}) k.
\end{align}
\end{proof}

\begin{proof}[Proof of Theorem~\ref{thm:knapsack}]
Let $X^*$ be the optimal solution and let $K = h(X^*)$. Let $X_i \lor X^* = X_i \lor x_1^* \lor \cdots \lor x_m^*$ such that each $x_j^*$ is admissible to $X_i \lor x_1^* \lor \cdots \lor x_{j-1}^*$.
Let $X^*_j = X_i \lor x_1^* \lor \cdots \lor x_j^*$.
\begin{align}
f(X^* \lor X_i) - f(X_i) 
&= \sum_j \left( f(X_j^*) - f(X_{j-1}^*) \right) \\
&= \sum_j \left( c(X_j^*) - c(X_{j-1}^*) \right) \frac{ f(X_j^*) - f(X_{j-1}^*) }{ c(X_j^*) - c(X_{j-1}^*) }  \\
& \le^{(*)} \sum_{j} (c(X_j^*) - c(X_{j-1}^*))\frac{ f(X_i \lor y_j^*) - f(X_i) + \delta}{c(X_j^*) - c(X_{j-1}^*)} \\
& \le \sum_{j} (c(X_j^*) - c(X_{j-1}^*))\frac{ f(X_i \lor y_j^*) - f(X_i)}{c(X_j^*) - c(X_{j-1}^*)} + \delta K\\
& \le^{(**)} \sum_{j} (c(X_j^*) - c(X_{j-1}^*))\frac{ f(X_i \lor y_j^*) - f(X_i)}{c(X_i \lor y_j^*) - c(X_i)} + \delta K
\end{align}
Here, in (*), we used the downward DR-submodularity of $f$: there exists $y_j^*$ such that $y_j^* \le \tilde x_j^*$, $\tilde x_j^* \in \cl{x_j^* \mid X_{j-1}^*}$, $y_j^* \in \adm{X_i}$, and $f(X_j^*) - f(X_{j-1}^*) \le f(X_i \lor y_j^*) - f(X_i)$.
Also, in (**), we used the order consistency of $c$: $c(X_j^*) - c(X_{j-1}^*) \ge c(X_i \lor y_j^*) - c(X_{i})$. 
%
By the greedy algorithm, we have
\begin{align}
... & \le \sum_{j} (c(X_j^*) - c(X_{j-1}^*))\frac{ f(X_{i+1}) - f(X_i)}{c(X_{i+1}) - c(X_i)} + \delta K \\
& \le \frac{ B }{c(X_{i+1}) - c(X_i)} \left( f(X_{i+1}) - f(X_i) \right) + K \delta
\end{align}
Here, we used the modularity of $c$: $\sum_j (c(X_j^*) - c(X_{j-1}^*)) = c(X^* \lor X_i) - c(X_i) = c(X^*) - c(X^* \land X_i) \le B$.
Therefore, by letting $\Delta_i = f(X^*) - f(X_i) - K \delta$
\begin{align}
	\Delta_{i+1} &\le \left(1 - \frac{c(X_{i+1}) - c(X_i)}{B}\right) \Delta_i \le \exp \left(-\frac{w(X_{i+1}) - c(X_i)}{B} \right) \Delta_i \\
    &\le \exp \left( \frac{-c(X_{i+1})}{B} \right) \Delta_0.
\end{align}
Therefore,
\begin{align}
	f(X_i) \ge (1 - e^{-c(X_i)/B}) f(X^*) - (1 - e^{-c(X_i)/B}) K \delta.
\end{align}
Let $t+1$ be the first step that the budget is exceeded.
\begin{align}
	f(X_t \lor a_{t+1}) \ge (1 - 1/e) f(X^*) - (1 - 1/e) K \epsilon.
\end{align}
Since $f(X_t \lor a_{t+1}) - f(X_t) \le f(b_{t+1})$ for some $b_{t+1}$ that is admissible to the bottom, by outputting the maximum of $X_t$ and the singletons, we obtain $((1 - 1/e)/2, K \epsilon (1 - 1/e)/2)$ approximation.
\end{proof}

\begin{proof}[Proof of Theorem~\ref{thm:nonmonotone}]

In the following analysis of Algorithm \ref{alg:nonmonotone}, the subscript $i$ means the objects after the $i$-th execution of the while loop.
\begin{lemma}
\label{lem:terminate}
During the algorithm, $A_i \le B_i$ holds for all iteration $i$.
In particular, Algorithm \ref{alg:nonmonotone} terminates in $h(\mathcal{L})$ iterations.
\end{lemma}
\begin{proof}[Proof of Lemma~\ref{lem:terminate}]
Lemma holds at $i = 0$.
Let us consider the general case.
If $A_i$ is updated, $A_{i+1} = A_i \lor a_i$ and $B_{i+1} = B_i$. Thus, by the induction and $a_i \le B_i$, the lemma holds.
Otherwise, $A_{i+1} = A_i$ and $B_{i+1} = \mathring{B}_i$. 
This holds by the definition of $\mathring{B}_i$.
\end{proof}

\begin{lemma}
\label{lem:nonmonotonesum}
Let $\alpha_i = f(A_{i-1} \lor a_i) - f(A_{i-1})$, $\beta_i = f(\mathring{B}_i) - f(B_{i-1})$. 
Then $\alpha_i + \beta_i \ge 0$ for all $i$.
\end{lemma}
\begin{proof}[Proof of Lemma~\ref{lem:nonmonotonesum}]
By the downward DR-submodularity,
\begin{align}
  f(B_{i-1}) - f(\mathring{B_{i}}) \le \max_{b: \mathring{B}_{i} \lor b = B_{i-1}} \min_{a \le b} f(A_{i-1} \lor a) - f(A_{i-1}).
\end{align}
Any $a$ in the above minimum satisfies $a \le B_{i-1}$ because of $a \le b$ and $b \le B_{i-1}$. 
By the definition of $a_{i}$, the right-hand-side is bounded by $f(A_{i-1} \lor a_i) - f(A_{i-1})$.
Combining these two inequality yields what we want to prove.
\end{proof}

Let $\mathrm{OPT}$ be an optimal solution and $\mathrm{OPT}_i = (\mathrm{OPT} \lor A_i) \land B_i$. 

\begin{lemma}
\label{lem:damage}
\begin{align}
\label{eq:damage}
f(\mathrm{OPT}_{i-1}) - f(\mathrm{OPT}_{i}) \le  f(A_i) + f(B_i) - f(A_{i-1}) - f(B_{i-1}) + \delta.
\end{align}
\end{lemma}
\begin{proof}[Proof of Lemma~\ref{lem:damage}]
In the following, we call $f(\mathrm{OPT}_{i-1}) - f(\mathrm{OPT}_{i})$ as damage and $f(A_i) + f(B_i) - f(A_{i-1}) - f(B_{i-1})$ as gain.
By the property of the algorithm, the gain is $\max\{\alpha_i, \beta_i\}$.
By the Lemma \ref{lem:nonmonotonesum}, the gain is always non-negative.

We first suppose that $\alpha_i \ge \beta_i$, i.e., $A_i$ is updated.
Because of the modular law, we have $\mathrm{OPT}_i = (\mathrm{OPT} \land B_i) \lor A_i$.
Hence, $\mathrm{OPT}_i = \mathrm{OPT}_{i-1} \lor a_i$.
This means that $0 \le h(\mathrm{OPT}_i) - h(\mathrm{OPT}_{i-1}) \le 1$.
If $h(\mathrm{OPT}_i) = h(\mathrm{OPT}_{i-1})$, then 
$\mathrm{OPT}_{i} = \mathrm{OPT}_{i-1}$ and the damage is zero.
If not, the upward DR-submodularity implies
\begin{align}
  f(\mathrm{OPT}_{i}) - f(\mathrm{OPT}_{i-1}) \ge \max_{b \ge a_i} \min_{\mathring{B}} ( f(B_{i-1}) - f(\mathring{B}) ) - \delta,
\end{align}
where $b \in \adm{\mathring{B}}$, $\mathring{B} \lor b = B_{i-1}$, and $\mathrm{OPT}_{i-1} \le \mathring{B}$.
In the definition of $\mathring{B}_i$, the variable $\mathring{B}$ in the algorithm runs over larger set than the inner minimum in the above since $A_{i-1} \subset OPT_{i-1}$ because of the fact $\mathrm{OPT}_i = (\mathrm{OPT} \land B_i) \lor A_i$.
Hence $\min_{\mathring{B}} f(B_{i-1}) - f(\mathring{B}) \ge - \beta_i$.
This means that the damage is bounded by $\beta_i$.

Next, we suppose that $\alpha_i \le \beta_i$, i.e., $B_i$ is updated.
By the modularity of the height, we have
\begin{align}
  h(\mathrm{OPT} \land B_{i-1}) &= h(\mathrm{OPT}) + h(B_{i-1}) - h(\mathrm{OPT} \lor B_{i-1} ), \\
  h( \mathrm{OPT} \land \mathring{B}_{i} ) &= h(\mathrm{OPT}) + h(\mathring{B}_{i}) - h(\mathrm{OPT} \lor \mathring{B}_{i} ).
\end{align}
By subtracting these two inequality, we have
\begin{align}
  0 \le h( \mathrm{OPT} \land B_{i-1} ) -  h( \mathrm{OPT} \land \mathring{B}_{i} ) \le 1,
\end{align}
since $\rank{\mathring{B}_i} + 1 = \rank{B_{i-1}}$ and $\mathrm{OPT} \lor \mathring{B}_{i} \le \mathrm{OPT} \lor B_{i-1}$.
If $\rank{ \mathrm{OPT} \land B_{i-1} } -  \rank{ \mathrm{OPT} \land \mathring{B}_{i} } = 0$, then $\mathrm{OPT} \land B_{i-1} =  \mathrm{OPT} \land \mathring{B}_{i}$.
This means that the damage is zero.
Otherwise, the downward DR-submodularity implies
\begin{align}
	f(\mathrm{OPT}_{i-1}) - f(\mathrm{OPT}_{i}) \le \max_{b} \min_{a'} ( f(A_{i-1} + a') - f(A_{i-1}) ) + \delta,
\end{align}
where $b$ runs over $b \in \adm{\mathrm{OPT}_i}, \mathrm{OPT}_{i-1} = \mathrm{OPT}_i \lor b$, and $a' \le b$ runs over $a' \in \adm{A_{i-1}}$.
In the definition of $a$ in the algorithm, $a$ runs over larger set than the inner minimum in the above since $a' \le B_{i-1}$ because $a \le b$ and $b \le B_{i-1}$.
Therefore, the damage is bounded by $\alpha_i$.
\end{proof}
By summing up the inequality~\eqref{eq:damage}, we have
\begin{align}
	f(\mathrm{OPT}) - f(\mathrm{ALG}) \le 2 f(\mathrm{ALG}) - f(\bot) - f(\top) + \rank{\mathcal{L}} \delta.
\end{align}
By the nonnegativity of $f$, the theorem is proved.
\end{proof}

\section{Experiment}
We check the difference between (\ref{eq:pca}) and (\ref{eq:generalizedpca}) by a numerical experiment.
In this experiment, we used the data on $\mathbb{R}^3$ whose coordinates are denoted by $(x_1, x_2, x_3)$.
We generate 1,000 data vectors $u_i$ ($i = 1, \ldots, 1,000$) each of which independently follows the identical Gaussian mixture distribution given by
\begin{align}
  p(x) = q \mathcal{N}(0, \Sigma_1) + (1 - q) \mathcal{N}(0, \Sigma_2),
\end{align}
where $q = 0.95$. 
Here, $\mathcal{N}(0,\Sigma)$ represents the probability distribution function of the normal distribution with mean zero and covariance matrix $\Sigma$, and $\Sigma_1$ and $\Sigma_2$ are given by
\begin{align}
  \Sigma_1 &= \left(
    \begin{array}{ccc}
      1 & 0 & 0 \\
      0 & 0.1 & 0 \\
      0 & 0 & 0.3
    \end{array}
  \right),\\
  \Sigma_2 &= \left(
    \begin{array}{ccc}
      0.1 & 0 & 0 \\
      0 & 1 & 0 \\
      0 & 0 & 0.3
    \end{array}
    \right).
\end{align}
The generated data is plotted in Figure~\ref{fig:experiment1}.
Our task is to find a two-dimensional subspace that captures the characteristics of this data.

\begin{figure}[t]
      \begin{tabular}{c}
      \begin{minipage}{0.33\hsize}
        \begin{center}
          (a)
          \includegraphics[width=4.5cm, clip]{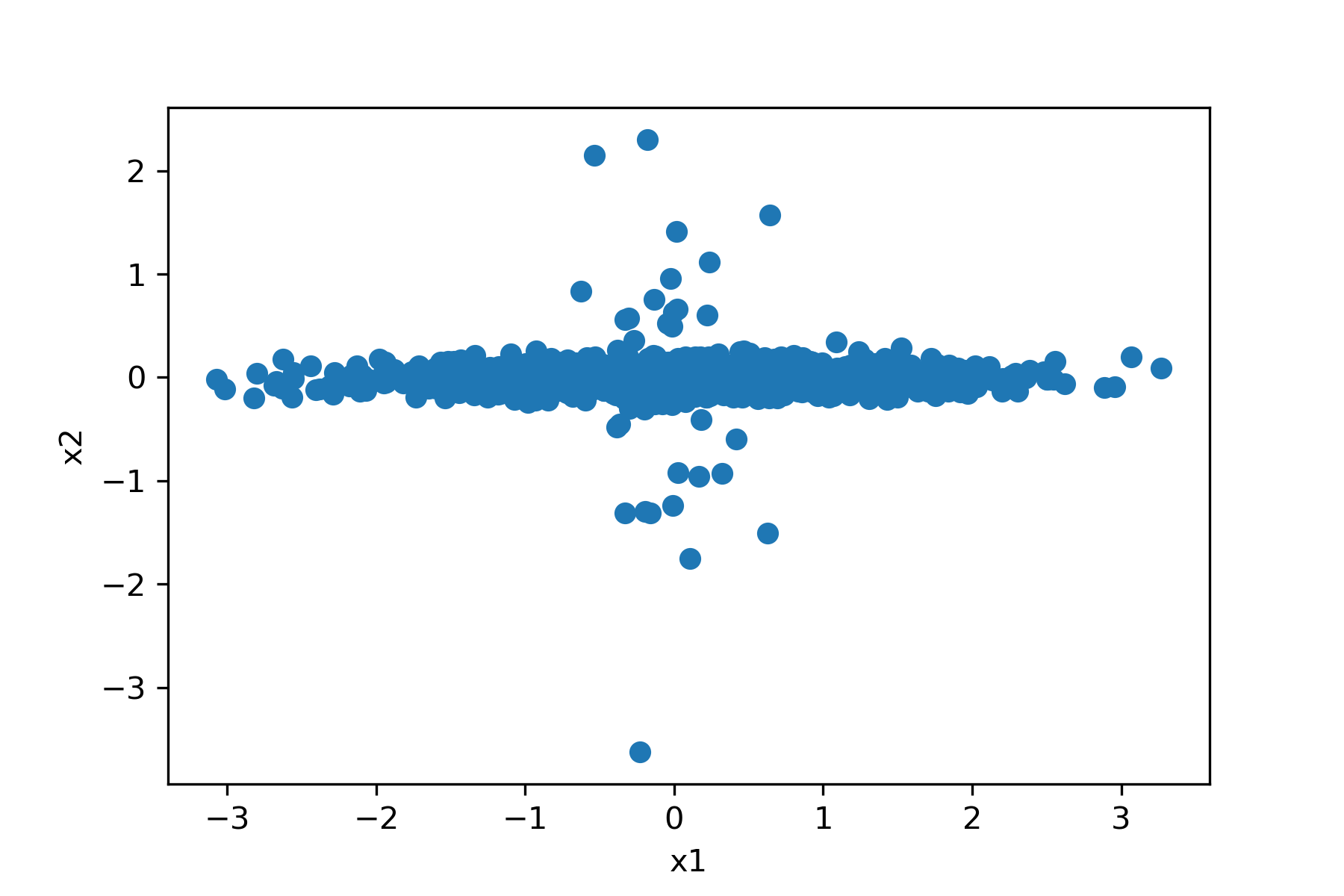}
        \end{center}
      \end{minipage}

      \begin{minipage}{0.33\hsize}
        \begin{center}
          (b)
          \includegraphics[clip, width=4.5cm]{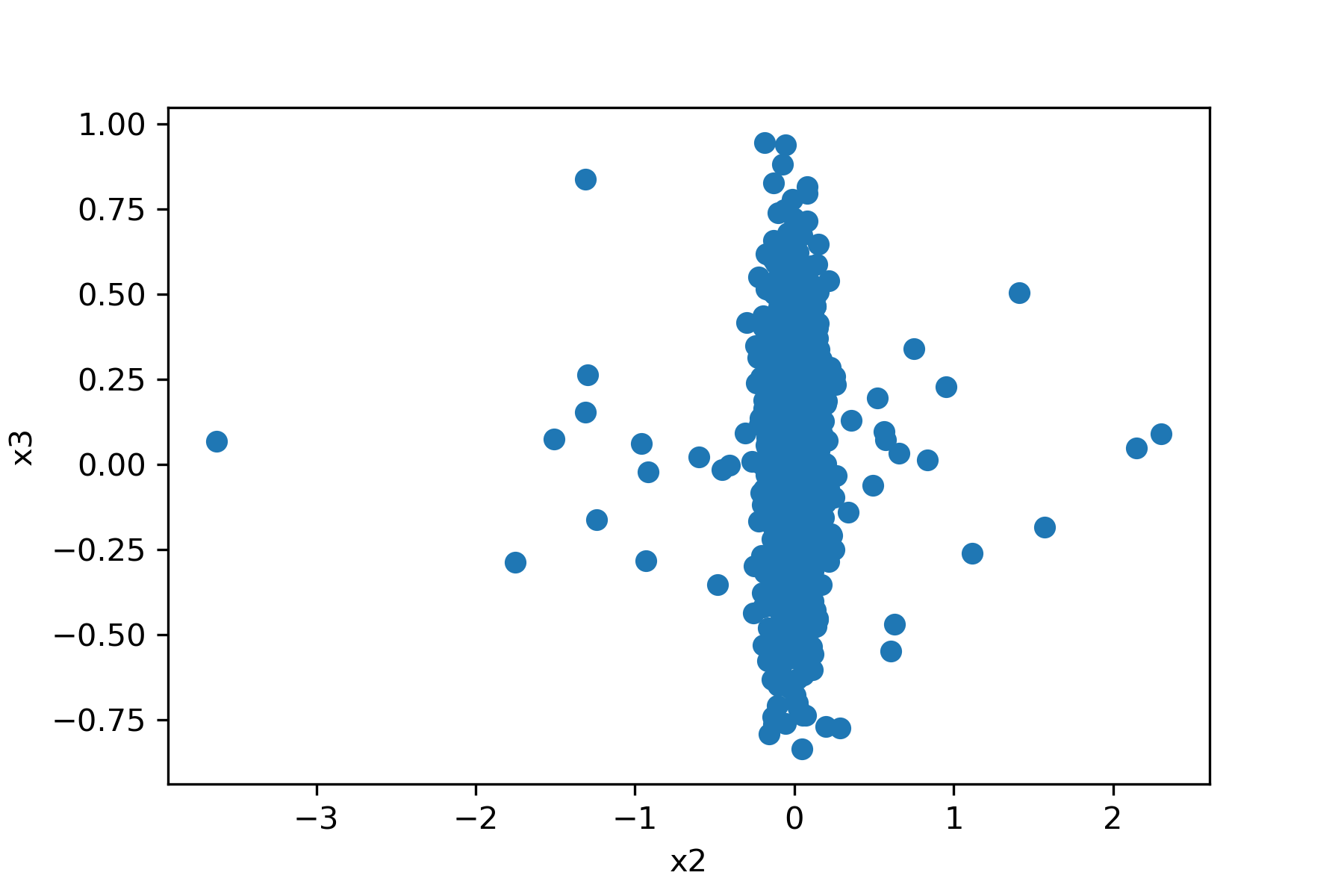}
        \end{center}
      \end{minipage}
      
      \begin{minipage}{0.33\hsize}
        \begin{center}
          (c)
          \includegraphics[clip, width=4.5cm]{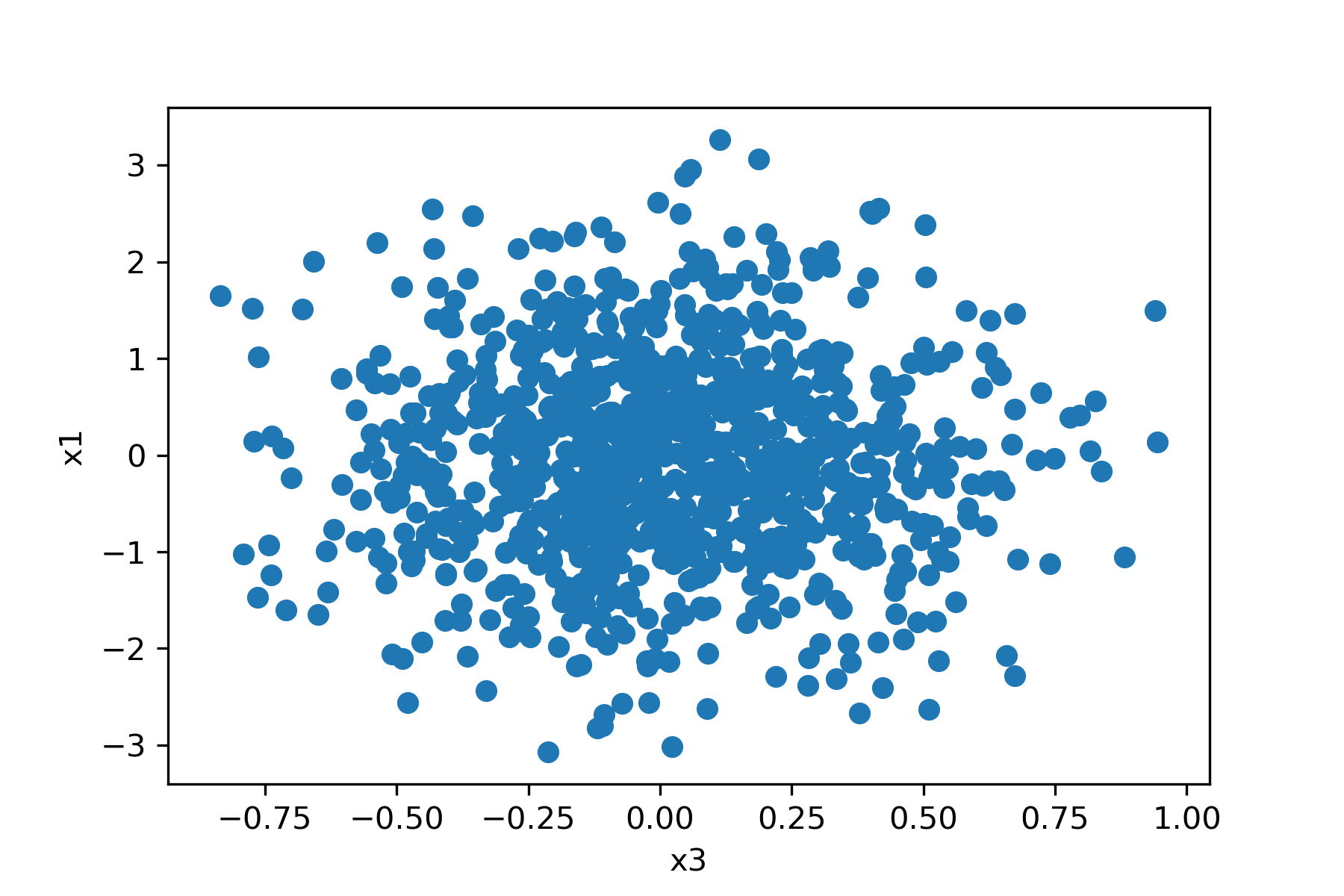}
        \end{center}
      \end{minipage}
    
      \end{tabular}

    \caption{The scatter plots of the data $\{u_i\}_{i \in I}$ used in the experiment. The figure (a) shows the projection of the data to $x_1$-$x_2$ plane, (b) shows to $x_2$-$x_3$ plane, and (c) shows $x_3$-$x_1$ plane.}
    \label{fig:experiment1}
\end{figure}

Since the data is generated from the mixture of two Gaussian distributions in which the first one spreads in $x_1$ direction and the second one spreads in $x_2$ direction, it is natural to we expect to find $x_1$--$x_2$ plane.
However, the ordinal principal component analysis yields an unexpected result as follows.
The first principal component is given by $( 0.9999,  0.0036,  0.0026)$, which represents the $x_1$ axis, but the second component is given by $(-0.0025, -0.0219,  0.9997)$, which represents the $x_3$ axis. 
Thus, they spans $x_1$--$x_3$ plane.
The reason for this unexpected result is that the first class is dominant compared to the second class in the data.
Thus, the ordinary principal component analysis yields the principal components of the first class regardless of the second class.

Now we apply the generalized principal component analysis (\ref{eq:generalizedpca}) to the data.
The function $\rho_i(t)$ is defined by
\begin{align}
  \rho_i(t) = \begin{cases}
    t & (t \le 0.01 \times \|x_i\|^2  ), \\
    0.1 \times (t - 0.01) + 0.01 & (\text{otherwise}).
  \end{cases}
\end{align}
Then, the generalized principal component analysis yields the expected result as follows:
The first component is $(-0.9999, -0.0065, -0.0017)$, which represents the $x_1$ axis, and the second component is $(0.0054,  0.9631, -0.2687)$, which represents the $x_2$ axis.
Thus, they successfully spans $x_1$--$x_2$ plane.
This shows an example that the generalized principal component analysis is more suitable than the ordinal principal component analysis.

The details of the implementation are as follows.
All algorithms are implemented in Python 3.
We generated the data by numpy.random.normal and np.random.binomial.
The PCA was computed by sklearn.decomposition.PCA.
The greedy choice of the vector in the generalized principal component analysis was done by scipy.optimize.differential\_evolution and scipy.optimize.brute.
Both methods yields the almost same result.
In scipy.optimize.brute, the searched grids are located on  $[0,1] \times [-1, 1] \times [-1, 1]$ with width 0.025.



\end{document}